\documentclass[12pt]{article}
\expandafter\let\csname equation*\endcsname\relax

\expandafter\let\csname endequation*\endcsname\relax

\usepackage{amsmath}
\usepackage{amsthm}
\usepackage{xcolor}
\usepackage{amssymb}
\usepackage{mathtools}
\usepackage{thmtools}
\usepackage{hyperref}
\usepackage{cleveref}
\newcommand\Mycomb[2][^n]{\prescript{#1\mkern-0.5mu}{}C_{#2}}
\def\bn#1{\mathbf{#1}}
\newcommand{\Set}[2]{%
	\{\, #1 \mid #2 \, \}%
}
\newcommand{\df}{\stackrel{def}{=}}

\newtheorem{lemma}{Lemma}
\newtheorem{proposition}{Proposition}
\theoremstyle{definition}
\newtheorem{littheorem}{Literature Theorem}[section]
\newtheorem{definition}{Definition}[section]
\newtheorem{remark}{Remark}[section]
\DeclareMathOperator*{\argmin}{arg\,min}

\title{Compact Representation of  $n^{th}$ order TGV}

\author{Manu Ghulyani \& Muthuvel Arigovindan}
\begin{document}
\maketitle

{Department of Electrical Engg., Indian Institute of Science, Bengaluru-12, Karnataka, India.}\\
{manug@iisc.ac.in \& mvel@iisc.ac.in }
\vspace{10pt}

\begin{abstract}
Although regularization methods based on derivatives are favored for their robustness and computational simplicity, research exploring higher-order derivatives remains limited. This scarcity can possibly be attributed to the appearance of oscillations in reconstructions when directly generalizing TV-1 to higher orders (3 or more). Addressing this, Bredies et. al introduced a notable approach for generalizing total variation, known as Total Generalized Variation (TGV). This technique introduces a regularization that generates estimates embodying piece-wise polynomial behavior of varying degrees across distinct regions of an image.Importantly, to our current understanding, no sufficiently general algorithm exists for solving TGV regularization for orders beyond 2 (i.e., $\geq 3$). This is likely because of two problems: firstly, the problem is complex as TGV regularization is defined as a minimization problem with non-trivial constraints, and secondly, TGV is represented in terms of tensor-fields which is difficult to implement. In this work we tackle the first challenge by giving two simple and implementable representations of $n^{th}$ order TGV \end{abstract}

%
%
%
%
%

\section{Notation and Preliminaries}\label{ch4:sec:notation}
\begin{enumerate}
	\item(Permutation) In this work, $\pi :\{1,2,3,...,k\}\rightarrow \{1,2,3,...,k\}$ is an (invertible) map called as a permutation. With this definition we can define a map $f_{\pi}:\mathbb R^{k}\rightarrow \mathbb R^k$, such that for  any $\mathbf v = (v_1,v_2,...,v_k)^T\in \mathbb R^k$,  $f_{\pi}(\mathbf v)\df ({ v}_{\pi(1)},{ v}_{\pi(2)},...,{ v}_{\pi(k)})^T$.  For example, let $k=3$, and  $\pi$ is defined such that: $\pi(1)=2, \pi(2)=3$ and $\pi(3)=1$. Then, for any $\bn v=(v_1,v_2,v_3)^T,f_{\pi}(\bn v)=(v_{\pi(1)}, v_{\pi(2)}, v_{\pi(3)})^T=(v_2,v_3,v_1)^T.$ We denote the set of of permutations of k-letters as $S_k.$  
	\item (Binary representation)  In this work, $b:\{0,1,2,3,...,2^k-1\}\rightarrow \{0,1\}^k$ gives the $k$ letter binary code i.e, $b$ gives a vector of $0s$ and ones for any non-negative integer less than $2^k$, and $b^{-1}$ is its inverse which returns a non-negative integer for any vector in $\{0,1\}^k.$ For example, if $k=3$, then $b(5)=[1,0,1]^T$ and similarly, $b^{-1}([0,1,1]^T)=3.$
	\item (symmetric index vector) For any non-negative integer $j\le 2^{k}-1,$ we can define an index vector $\bn t^{j,k}$ as $\bn t^{j,k}=[b^{-1}f_{\pi_1}b(j),b^{-1}f_{\pi_2}b(j),...,b^{-1}f_{\pi_{k!}}b(j)].$ Here, $\pi_1,..,\pi_{k!}$ are the elements of $S_k$ in any fixed order. As an example, consider $k=2$ and $j=2$. In this case there are two permutations, $\pi_1$ is the identity map   and $\pi_2$ is defined such that: $\pi_2(1)=2$ and $\pi_2(2)=1.$ Therefore, $\bn t^{2,2}=[2,1]^T$.
	\item 
	We define a linear operator, $\Pi^{(k)}:\mathbb R^{N\times 2^k}\rightarrow \mathbb R^{N\times 2^k}$. For any $ P\in \mathbb R^{N\times 2^k},$ the $(i,j)^{th}$ element of $(\Pi^{(k)}(P))$ is given as $$(\Pi^{(k)}(P))_{i,j}\df \frac{P_{i,\bn t^{j,k}_1}+P_{i,\bn t^{j,k}_2}+...+P_{i,\bn t^{j,k}_{k!}}}{k!}.$$Here, $\bn t^{j,k}=[\bn t^{j,k}_1,\bn t^{j,k}_2,...,\bn t^{j,k}_{k!}]^T$ is the symmetric index vector as defined above. It can be observed that  the $j^{th}$ column of $(\Pi^{(k)}P)$  is the   sum the columns of the  matrix $P$ given by the column indices $\{b^{-1}f_{\pi}(b(j))\}_{\pi\in {S_k} }$ for different permutations $\pi$. 
	
	\item Consider a scanned image containing $N$ ordered pixels, then  $\mathbf D_x$ and $\mathbf D_y$ are matrix-equivalent of discrete derivatives in $x$ and $y$ directions respectively applied directly on the scanned image. For example, let $\mathbf u$ be an image having $N$ pixels, then $\mathbf D_x \bn u$ is the derivative image in $x$ direction, and similarly $\bn D_y
	\bn u$ in y direction. As  we consider images in scanned form  we have  $\bn u \in \mathbb R^N$; therefore, $\bn D_x$ and $\bn D_y$ are $N\times N$ block circulant matrices with circulant blocks (BCCB), and the multiplication of these matrices with any vector in $\mathbb R^N$ is implemented by $2D$ convolution with filters $(-1,1)$ and $(-1,1)^T$ respectively.
	\item The iterated derivative operator $\mathcal D_k:\mathbb R^{N \times (k+1) }\rightarrow \mathbb R^{N \times (2k+2) }$, for any $(\bn y_0,...,\bn y_k)\in \mathbb R^{N\times (k+1)} \text{ is defined as} ,$\\$$\mathcal D_k([\mathbf y_0, \mathbf y_1,...,\mathbf y_k])\df[\mathbf D_x \mathbf y_0,\mathbf D_y \mathbf y_0,\mathbf D_x \mathbf y_1,\mathbf D_y \mathbf y_1,...,\mathbf D_x \mathbf y_k,\mathbf D_y \mathbf y_k].$$
	\item 	We need a scaling linear operator ($\mathcal A_k$) to define the compact form of TGV, $\mathcal       A_k:\mathbb R^{N \times (2k+2) }\rightarrow \mathbb R^{N \times (k+2) }$,   
	such that, $\mathcal A_k(\mathbf z)\df \mathbf z\cdot M_k$, \text{where} $$M_k=\begin{bmatrix} &1 & 0 &0&...&0&0\\&0&\frac{\sqrt{\Mycomb[k]{1}}}{\sqrt {\Mycomb[k+1]{1}}}&0&...&0&0\\ &0&\frac{\sqrt{\Mycomb[k]{0}}}{\sqrt {\Mycomb[k+1]{1}}}&0&...&0&0\\&0&0&\frac{\sqrt{\Mycomb[k]{2}}}{\sqrt {\Mycomb[k+1]{2}}}&...&0&0\\
	&0&0&\frac{\sqrt{\Mycomb[k]{1}}}{\sqrt {\Mycomb[k+1]{2}}}&...&0&0\\
	&\vdots&\vdots&\vdots&...&\vdots&\vdots\\
	&0&0&0&...&\frac{\sqrt{\Mycomb[k]{k}}}{\sqrt {\Mycomb[k+1]{k}}}&0\\
	&0&0&0&...&\frac{\sqrt{\Mycomb[k]{k-1}}}{\sqrt {\Mycomb[k+1]{k}}}&0\\
	&0&0&0&...&0&1\\
	\end{bmatrix}.$$ 
	\item Let $\mathcal A$ be a linear operator, then $\mathcal N(\mathcal A)$ denotes the null space of $\mathcal A$ and $\mathcal R(\mathcal A)$ denotes the range space of $\mathcal A$.
	\item We need the definition of proximal operator to define the image restoration algorithm. For any lower-semi-continuous, convex and closed function $f:\mathbb R^N\rightarrow \mathbb R$. The proximal of $f()$  is given as:
	$$prox_{f}(\bn z)\df \inf_{x}f(\bn x)+\frac{1}{2}\|\bn x-\bn z\|^2.$$
	\item \
	The mixed norm ($\|\cdot\|_{1,2}$) for any $\mathbf M \in \mathbb R^{N\times p} $ is defined as:$$\|\mathbf M\|_{1,2}=\sum_{j=1}^{N}(\sum_{i=1}^{p}\mathbf M_{j,i}^2)^{1/2}.$$
	\item Proximal of $\|\|_{1,2}$ norm: Consider $\bn A\in \mathbb R^{M\times N}$ then \begin{align}prox_{\|\cdot\|_{1,2}}(\bn A)&=\argmin_{\bn B}[\|\bn B\|_{1,2}+\frac{1}{2}\|\bn B -\bn A\|^2]\\&=\argmin_{\bn B}\sum_{i=1}^M\big[\sum_{j=1}^N((\bn B-\bn A)_{i,j})^2+(\sum_{j=1}^N \bn B_{i,j}^2)^{\frac{1}{2}}\big]
	\end{align}
	The above optimization is separable in $i$ (row index). Hence, $i^{th}$ row  of the minimizer of the above expression $\bn A^*$ which is same as $prox_{\|\cdot\|_{1,2}}(\bn A)$  can be given as: 
	\begin{align}
	&\bn A^*_{i,:}=prox_{\|\cdot\|_2}(\bn A_{i,:})\\&=\max(0,1-\frac{1}{\sqrt{\sum_j \bn A_{i,j}^2}})\bn A_{i,:}.
	\end{align}
	Therefore, the proximal can be obtained by performing proximal of $l_2$ norm on each row.
\end{enumerate}

\section{Introduction to Total Generalized Variation (TGV)}
One of the most important regularizations for image restoration is Total Variation (TV) \cite{rudintv}. TV penalizes the sum  of $l_2$ norm of the gradient of the image. Therefore, the resultant image is piece-wise constant. In our notation,
$$TV^{1}(\bn u)=\|\mathcal D_{0}\bn u\|_{1,2}=\|(\mathbf D_x  \bn u ,\mathbf D_y \bn u)\|_{1,2}.$$ The above concept can be extended to second-order derivatives also as follows: \begin{align}&TV^{2}(\bn u)=\|\mathcal  D_1 ( \mathcal D_{0}(\bn u))\|_{1,2}=\|(\mathbf D_x \mathbf D_x \bn u ,\mathbf D_y \mathbf D_x \bn u ,\mathbf D_x \mathbf D_y \bn u,\mathbf D_y \mathbf D_y \bn u)\|_{1,2}. \end{align} 
Analogously, one can define $TV^n$. Research has indicated, particularly for 1-dimensional signals, that employing the aforementioned $TV^n$ regularization yields a solution that takes the form of a linear combinations of polynomials with a fixed degree of $n-1$ \cite{unser_represent}. However, images are typically better characterized by piece-wise smooth polynomials, which might possess varying degrees across different regions, rather than adhering strictly to a fixed polynomial degree of $n-1$ across the entire image. Consequently, there arises a necessity for a more adaptable and robust extension of the Total Variation (TV) concept.

TGV, an influential contribution by \cite{tgv}, demonstrates how the restored image is represented as a linear combination of polynomials with varying degrees across distinct image sections. In other words, TGV has the capability to generate solutions that manifest as piece-wise polynomials of diverse degrees in different parts of the image. For a thorough mathematical exploration, readers can refer into the details presented in \cite{tgv}.

Original formulations of TGV are rooted in a continuous domain rather than a discrete pixel grid. In this context, the formulation of second-order TGV, denoted as $TGV^2$, can be expressed as follows:
\begin{align}TGV^2( u)=\\ &  \nonumber \inf_{v_1,v_2\in \mathcal S_\Omega}\underbrace{\int_{\Omega}\big[\big(\frac{\partial u}{\partial x}-v_1 \big)^2+\big(\frac{\partial u}{\partial y}-v_2 \big)^2\big]^{1/2}dxdy}_{A}+\\&\underbrace{\int_{\Omega}\big[\big(\frac{\partial v_1}{\partial x} \big)^2+\big(\frac{\partial v_2}{\partial y} \big)^2+\frac{1}{2}\big(\frac{\partial v_1}{\partial y}+ \frac{\partial v_2}{\partial x}\big)^2\big]^{1/2}dxdy}_B.\end{align}
In the given equation, observe that component "A" effectively matches the partial derivatives of $v_1$ and $v_2$, while component "B" further refines the matched partial derivatives ($v_1$ and $v_2$) by incorporating a second-order total variation regularization. This pattern can be extended generally by iteratively fitting the $(n-1)$-th order derivative and then applying regularization using the $n$-th order derivative. The original continuous TGV 
formulation \cite{tgv} (eq. 3.6) is given as \cite{tgv}:\begin{align}\label{tgv_cont}\textit{TGV}^n(g)=\inf_{ u_i\in \mathbb S^{(i)}_{\Omega},u_0=g,u_n=0}\sum_{i=0}^{n-1}\| \epsilon u_i-u_{i+1}\|.\end{align} Here, $\mathbb S_{\Omega}^{i}$ is the space of $i$ dimensional symmetric tensor fields on $\Omega$ having bounded deformation and $\epsilon$ is the symmetric derivative. Compared with the formulation given in this paper,   $\epsilon$ is analogous to $\Pi^{(i)}\circ \mathcal D_{2^i-1}$. It can be noted that $\Pi^{(i)}\circ \mathcal D_{2^i-1}$ varies with $i$ (the derivative order) while $\epsilon$ does not as  the authors have overloaded the operator for various derivative orders, but we have defined a different operator for each derivative order ($i$). This does not create any difference in the formulation. Also, $\mathcal S_{\Omega}$ is the set of  functions on the set $\Omega$ having bounded deformation \cite{tgv}. For details, the reader can refer to \cref{ch4:sec:proof_rep}.\\ For $n=2$, the above continuous formulation can be written in discrete form as :
$$TGV^{2}(\bn u)=\inf_{\bn p \in \mathbb R^{N\times 2}}\alpha_1\|\mathcal D_{0} \bn u-\bn p\|_{1,2}+\alpha_0 \|\Pi^{(2)}\mathcal D_1 \bn p\|.$$
In the given equation, the parameters $\alpha_1$ and $\alpha_0$ control the regularization. When $\alpha_1$ becomes exceedingly large, the regularization effectively becomes equivalent to $TV^2$. Conversely, when $\alpha_0$ tends towards infinity, the regularization behaves like $TV^1$. This intriguingly means that the Total Generalized Variation ($TGV$) approach has the flexibility to emulate both $TV^1$ and $TV^2$ regularization based on different choices of parameters, all while adapting spatially due to the variable $\bn p$. This is because, in regions where $\bn p=\bn 0$, the regularization takes on the characteristics of $TV^1$, while in regions where $\bn p=\mathcal D_0\bn u$, the regularization functions as $TV^2$.

The  continuous version of TGV (\cref{tgv_cont}) needs to be discretized for practical implementation. Although TGV is originally defined for continuous images (smooth functions on $\mathbb R^2$) and tensor fields; tensors are not necessary to describe the discretized TGV (as we show in this work). In this work, we give matrix-based formulation of TGV. Removing this additional hurdle to learning the tensor machinery makes this accessible to a wider audience.

There  are many algorthms that solve image restoration with TGV regularization.
Originally, TGV was proposed by \cite{tgv}. In this work, they gave a primal-dual based algorithm that could solve the image-denoising problem up to third-order TGV. They also gave many structural and theoretical properties of TGV. In another work \cite{tgv_numerics}, the authors gave a primal-dual scheme to perform image decompression, zooming, and image reconstruction using second-order TGV. MRI reconstruction with second-order TGV was proposed in \cite{TGV2}.
\subsection{Tensor-free and Compact representation of TGV}
It is important to note that because of the complexity of TGV for orders ($\ge 3$), higher-order TGV is mainly unexplored for non-trivial inverse problems.  As 2-tensors can be easily understood as vector images,  discretizing and implementing TGV-2 is simple and easier to implement, hence, most works focus only on TGV-2. Only the work by \cite{tgv} focuses on TGV order 3 for the denoising problem. To the best of our knowledge, the work described in this article is the first work that gives a  general algorithm to solve a linear inverse problem for any order of TGV. Now, we give a theorem that established the direct matrix version of TGV.
\begin{restatable}{theorem}{tgvdir}\label{ch4:th:tgvdir} (Direct Tensor-free Representation of TGV) Let $\bn g $ be an image having N pixels. Then, the total generalized variation in discrete form is given as:
	\begin{equation}\label{eq:tgv_direct}
	TGV^n(\mathbf g)=\inf_{\mathbf p_0=\mathbf g,\mathbf p_n=\mathbf 0,\mathbf p_i\in \mathbb R^{N\times (2^i)},\bn p_i\in \mathcal R(\Pi^{(i+1)})}\sum_{i=0}^{n-1}\alpha_{n-i-1}\|  \Pi^{(i+1)} \mathcal  D_{2^i-1} \mathbf p_i-\mathbf p_{i+1}\|_{1,2}.
	\end{equation}
	Here, $\alpha_0,...,\alpha_{n-1}$ are the regularization parameters and $ \Pi^{(i+1)}$ and $\mathcal D_{2^i-1}$ are as given in the \cref{ch4:sec:notation}. 
	
\end{restatable}
As an  example:  $TGV^3$ is given as : $$TGV^{3}(\bn u)=\inf_{\bn p_1\in \mathbb R^{N\times 2}, \bn p_2\in  \mathbb R^{N\times 4}, \bn p_2\in \mathcal R(\Pi^{(2)})}\alpha_2\|\mathcal D_{0} \bn u-\bn p_1\|_{1,2}+\alpha_1 \|\Pi^{(2)}\mathcal D_1 \bn p_1-\bn p_2\|+\alpha_0\|\Pi^{(3)}\mathcal D_3\bn p_2\|.$$
The aforementioned direct form is impractical for implementation due to two primary reasons. Firstly, it entails matrices (as minimization variables) of size ($\mathbb R^{N\times 2^{n-1}}$) that grows exponentially with the TGV order ($n$). Secondly, it imposes a constraint that $\mathbf{p}_i \in \mathcal{R}(\Pi^{(i)})$, further adding complexity to the problem. To address these challenges, we present a more concise expression for TGV:\begin{restatable}{theorem}{tgvd}(Compact tensor-free representation of TGV)\label{ch4:th:rep}
	Let $\bn g $ be an image having N pixels. Then, the expression given in the \cref{ch4:th:tgvdir} can also be written as::
	\begin{equation}\label{eq:tgv_compact}
	TGV^n(\mathbf g)=\inf_{\mathbf p_0=\mathbf g,\mathbf p_n=\mathbf 0,\mathbf p_i\in \mathbb R^{N\times (i+1)}}\sum_{i=0}^{n-1}\alpha_{n-i-1}\| \mathcal A_i \mathcal D_i \mathbf p_i-\mathbf p_{i+1}\|_{1,2}.
	\end{equation}
	Here, $\alpha_0,...,\alpha_{n-1}$ are the regularization parameters and $\mathcal A_i$ and $\mathcal D_i$ are as given in the \cref{ch4:sec:notation}. 
\end{restatable}
For proof,  see \cref{ch4:sec:proof_rep}. 
Note that both problems encountered in the direct formulation are addressed: (1) the issue of optimization variables ($\mathbf{p}_i$'s) growing in size exponentially has been solved. More precisely, now $\bn p_i$ are of the size $\mathbb R^{N\times (i+1)}$ instead of  $\mathbb R^{N\times 2^{i-1}}$, and (2) the constraint on the variables $\mathbf{p}_i$'s is now absent from the formulation. These challenges are eliminated through the reformulation provided in the theorem above.

Furthermore, it's important to highlight that readers aiming to understand and implement this formulation do not need any additional background knowledge on tensors. This significantly enhances the accessibility of the method.
\begin{restatable}{remark}{eqtgv}\label{ch4:remark:dirrep}(Direct and Compact expressions of TGV are  equivalent) It can be noted that proof of the compact form is derived from the original TGV definition. But, one can arrive the compact form starting from the direct form with the help of the following statement.
	For any $\bn g \in \mathbb R^N$, $\inf\Set{{\sum_{i=0}^{n-1}}\alpha_{n-i-1}\|\Pi^{(i+1)}\mathcal D_{2^i-1}\bn u_i-\bn u_{(i+1)}\|_{1,2}}{\bn u_n=\bn 0,\bn u_0=\bn g ,\bn u_i\in \mathcal R(\Pi^{(i)}) \text{ for } i=1,...,n-1}= \inf\Set{{\sum_{i=0}^{n-1}}\alpha_{n-i-1}\|\mathcal A_i\mathcal D_{i}\bn p_i-\bn p_{i+1}\|_{1,2}}{\bn p_n=\bn 0,\bn p_0=\bn g ,\bn p_i\in \mathbb R^{N\times{(i+1)}}\text{ for }i =1,...,n-1}$ 
\end{restatable}
The proof of the above statement is deferred to the end.
As a result of the above remark one can use the compact representation of TGV for all purposes.

\section{ Proof of \cref{ch4:th:tgvdir} and  \cref{ch4:th:rep}}\label{ch4:sec:proof_rep}
\subsection{Preliminaries on tensors and tensor based TGV}
\begin{definition}(Tensors)\label{def:tensors}
	A function $\mathcal P: \underbrace{\mathbb R^2 \times \mathbb R^2 .....\times \mathbb R^2}_{k \ \ times} \rightarrow \mathbb R$ is called a k-tensor on $\mathbb R^2$ if it satisfies the following:
	\begin{enumerate}
		\item $\mathcal P (\mathbf{v}_1,\mathbf{v}_2,...,\alpha \mathbf{v}_i,...,\mathbf v_k)=\alpha \mathcal P (\mathbf{v}_1,\mathbf{v}_2,..., \mathbf{v}_i,...,\mathbf v_k)$ for any $\alpha \in \mathbb R$ and any $i \in \{1,2,...,k\}$. Here, $\mathbf v_i \in \mathbb R^2$ for any $i \in \{1,2,...,k\}$.
		\item $\mathcal P (\mathbf{v}_1,\mathbf{v}_2,...,\mathbf{v}_i+\mathbf w,...,\mathbf v_k)= \mathcal P (\mathbf{v}_1,\mathbf{v}_2,..., \mathbf{v}_i,...,\mathbf v_k)+\mathcal P (\mathbf{v}_1,\mathbf{v}_2,...,\mathbf w,...,\mathbf v_k)$ for any $\mathbf w \in \mathbb R^2$ and any $i \in \{1,2,...,k \}$.
	\end{enumerate}
\end{definition}

For example, the function $\mathcal F : \mathbb  R^2 \rightarrow \mathbb R$, such that $\mathcal F(\mathbf v)=\mathbf a^T \mathbf v$, where $\mathbf a$ is any fixed vector in $\mathbb R^2$, is a 1-tensor in $\mathbb R^2.$ Similarly, the  function $\mathcal G : \mathbb  R^2 \times \mathbb R^2 \rightarrow \mathbb R$, such that $\mathcal G(\mathbf v_1,\mathbf v_2)=\mathbf v_1^T \mathbf v_2$ is a 2-tensor in $\mathbb R^2.$ 
\begin{remark}(Space of $k-$Tensors-$\mathbb C^{(k)}$)
	
	It can be verified that the set of all k-tensors in $\mathbb  R^2$ constitutes a vector space.  We denote this vector space by $\mathbb C ^{(k)}.$\end{remark}
\begin{definition} (Tensor Product)
	Consider a k-tensor $\mathcal P$ and an l-tensor $\mathcal Q$. Then, the tensor product ($\otimes$) of $\mathcal P$ and $\mathcal Q$ is a k$+$l tensor which is  given as:$$(\mathcal P \otimes \mathcal Q)(\mathbf v_1,\mathbf v_2,.., \mathbf v_{k+l})=\mathcal P(\mathbf v_1,...,\mathbf v_k)\cdot\mathcal Q(\mathbf v_{k+1},...,\mathbf v_{k+l}).$$\end{definition}

\begin{definition}\textbf{(Permutation)} In this work, $\pi : \{1,2,3, \ldots, k\} \rightarrow \{1,2,3, \ldots, k\}$ is an (invertible) map known as a permutation. With this definition, we can define a map $f_{\pi} : \mathbb{R}^k \rightarrow \mathbb{R}^k$, such that for any $\mathbf{v} = (v_1, v_2, \ldots, v_k)^T \in \mathbb{R}^k$, $f_{\pi}(\mathbf{v}) \df (v_{\pi(1)}, v_{\pi(2)}, \ldots, v_{\pi(k)})^T$. For example, let $k=3$, and $\pi$ is defined such that: $\pi(1) = 2$, $\pi(2) = 3$, and $\pi(3) = 1$. Then, for any $\mathbf{v} = (v_1, v_2, v_3)^T$, we have $f_{\pi}(\mathbf{v}) = (v_{\pi(1)}, v_{\pi(2)}, v_{\pi(3)})^T = (v_2, v_3, v_1)^T$. We denote the set of all permutations of k-letters as $S_k.$
	
\end{definition}
\begin{definition}(Symmetric k-tensors) A k-tensor $\mathcal S$ is  symmetric if $\mathcal S(\mathbf v_1,\mathbf v_2,...,\mathbf v_k)=\mathcal S(\mathbf v_{\pi(1)},...,\mathbf v_{\pi(k)})$ for all $\pi \in S_k.$ Here, $\pi$ is any permutation of k letters. 
\end{definition}
It can be noted that the set of symmetric tensors is a sub-space of the space of k-tensors. We denote this sub-space by $\mathbb S^{(k)}.$ Since, the symmetric k-tensors forms a subspace a projection ($|||$) on the sub space of symmetric k-tensors can be defined.
\begin{definition}\label{def:sym}
	The projection $|||^{(k)}:\mathbb C^{(k)}\rightarrow \mathbb S^{(k)}$ is defined as: $$|||^{(k)}\mathcal P(\mathbf v_1,\mathbf v_2,...,\mathbf v_k)=\frac{1}{k!}\sum_{\pi\in S_k}\mathcal P(\mathbf v_{\pi(1)},...,\mathbf v_{\pi(k)}).$$\end{definition}
The above expression can be interpreted as an average of all $k!$ permutations.
We need the following theorem  regarding the basis of the space of k-tensors in $\mathbb R^2.$ 

\begin{littheorem} (Standard basis for the space of  $k-$tensors)
	Consider $\mathbf w_0=
	\begin{pmatrix}
	1 \\
	0
	\end{pmatrix}
	$ and $\mathbf w_1=
	\begin{pmatrix}
	0 \\
	1
	\end{pmatrix}
	.$ $\{\mathbf w_0,\mathbf w_1\}$ is basis of $\mathbb R^2$ and let $\mathbf\omega_0=\mathbf w_0^T$ and $\mathbf\omega_1=\mathbf w_1^T$ be the  corresponding dual basis. Here, $\omega_0$ and $\omega_1$ are linear operators on $\mathbb R^2.$ For example, if $\bn v=(v_1 \ v_2)^T$, then $\omega_0(\bn v)=\bn w_0^T\bn v=v_1$ and $\omega_1(\bn v)=\bn w_1^T\bn v=v_2.$Then, $\{\mathbf \omega_{i_1}\otimes\mathbf \omega_{i_2}\otimes...\otimes\mathbf \omega_{i_k}| \ \ i_p\in \{0,1\} \text{ for all  } p  \ \in \{1,2,...k\}\}$ is the basis of $\mathbb C^{(k)}$.
\end{littheorem}
As for example, it can be verified that \begin{align*}
\mathcal G(\mathbf v,\mathbf w)&=\bn v^T\bn w\\&= v_1 w_1+ v_2w_2\\&=(\omega_0(\bn v))(\omega_0 (\bn w))+(\omega_1(\bn v))(\omega_1 (\bn w))\\&=(\omega_0\otimes \omega_0+ \omega_1\otimes\omega_1)(\mathbf v,\mathbf w).\end{align*} 
\begin{remark}
	With this basis we can represent any k-tensor, $\mathcal P$ in $\mathbb C^{(k)}$ as a summation of the all possible $2^k$ basis vectors as $\mathcal P=\sum_{\bn i \in \{0,1\}^k}p_{\bn i}\pmb \omega_{\bn i}.$. Here, $\bn i=(i_1,i_2,...,i_k)$ is the vector index lying in the set $\{0,1\}^k,$ therefore each $i_j$ takes the value either $0$ or $1.$ Hence, there are $2^k$ coefficients ( namely $p_{\bn i}s$) corresponding to all $\bn i$ that are in the set $\{0,1\}^k.$ Here, $\pmb \omega_{\bn i}$ is defined as $\pmb \omega_{\bn i} \df \omega_{i_1}\otimes...\otimes \omega_{i_k}.$ 
\end{remark}
\begin{remark} Let $\bn i$ be any vector in $\{0,1\}^k,$ then the sum, $s(\bn i)$ is defined as $s(\bn i)\df \sum_{j=1}^k i_j$.
	\begin{littheorem}\label{cor:e} (Orthonormal basis for the space of symmetric tensors) An orthonormal basis for the subspace $\mathbb S^{(k)}$ can be 
		given by the set $\{\mathbf e_0^{(k)},\mathbf e_1^{(k)},...,\mathbf e_k^{(k)}\}$ \cite{tgv} where $$\mathbf e_i^{(k)}=\frac{1}{\sqrt{\Mycomb[k]{i}}}\sum_{\mathbf j \in \{0,1\}^{k},s(\bn j)=i} \pmb \omega_{\bn j} .$$
	\end{littheorem}
	
\end{remark}
\begin{remark}
	It can be deduced from the above theorem that the dimension of $\mathbb C^{(k)}$ is $2^k$. Since, any finite dimensional vector space is isomorphic to the euclidean space (of same dimension), we can conclude that $\mathbb C^{(k)}$ is  isomorphic to $\mathbb R^{2^k}$. Explicitly,  there is an isomorphism $\psi: \mathbb C^{(k)}\rightarrow \mathbb R^{2^k}$, such that for any $\mathcal P=\sum_{\bn i\in \{0,1\}^k}p_{\bn i}\pmb \omega_{\bn i}\in \mathbb C^{(k)}$,  $ \psi(\mathcal P)\df (p_{\bn i^{(1)}},p_{\bn i^{(2)}},...,p_{\bn i^{(2^k)}})\in \mathbb R^{2^k}$. Here, $\bn i^{(1)},...,\bn i^{(2^k)}$ are the elements of $\{0,1\}^k$ arranged in increasing order, i.e, $b(\bn i^{(1)})<b(\bn i^{(2)})<...<b(\bn i^{(2^k)})$.
\end{remark}
\begin{remark} (Formula Projection of a tensor given in standard basis on $\mathbb S^{(k)}$) In \cref{def:sym}  we gave the definition of the projection on the space of symmetric tensors, now we give a formula for the coefficient of the projected tensor (given in standard basis) corresponding to the basis vectors $\{\pmb \omega_{\bn  j}\}_{\bn j \in \{0,1\}^k}$. Consider an element in $\mathbb C^{(k)}$ that is represented in standard basis as $\mathcal Q=\sum_{\bn j\in \{0,1\}^k}q_{\bn j}\pmb \omega_{\bn j} \in \mathbb C^{(k)},$  then $(|||^{(k)}\mathcal Q)=\sum_{\bn j\in \{0,1\}^k}\big(\frac{1}{k!}\sum_{\pi \in S_k} q_{f_{\pi}(\bn j)}\big)\pmb \omega_{\bn j}$. Recall that $f_{\pi}$ permutes the elements of the vector $\bn j$ according to the permutation  map $\pi.$
	
\end{remark}
We also need the definition of norm in $\mathbb C^{(k)}.$ 
\begin{definition}(inner product)Recall that $\bn w_0=(1,0)^T$ and $\bn w_1=(0,1)^T$. 
	Consider two tensors $\eta$ and $\zeta$ in $\mathbb C^{(k)}$, then the inner $\langle \eta,\zeta\rangle$ product is given  by:
	$$\langle \eta,\zeta \rangle=\sum_{\mathbf i\in\{0,1\}^{k}}\eta(\mathbf w_{i_1},\mathbf w_{i_2},...,\mathbf {w_{i_k}}) \zeta(\mathbf w_{i_1},\mathbf w_{i_2},...,\mathbf {w_{i_k}}).$$ Consequently, we can define the norm of any k-tensor $\eta \in \mathbb C ^{(k)}$ as $$\|\eta\|=\sqrt{\langle \eta,\eta \rangle}.$$
\end{definition}
\begin{remark}(Norm of a tensor)
	The norm ($\|\cdot\|$) of any tensor $\alpha=\sum_{\mathbf i \in \{0,1\}^{k}}\alpha_{\bn i}\pmb \omega _{\bn i}$ can also be given as:
	$\|\mathbf \alpha\|=(\sum_{\mathbf i \in \{0,1\}^{k}}\alpha_{\bn i}^2)^{1/2}.$ To obtain this relation, one can use the identity that $\omega_{i_1}\otimes....\otimes\omega_{i_k}(\mathbf w_{j_1},\mathbf w_{j_2},...,\mathbf{w}_{j_k})=1$ if $ i_l= j_l$ for all $l\in \{1,2,..,k\}$ and  zero otherwise.
\end{remark}It can be verified that for any $\beta=\sum_{i=0}^{k}\beta_i \mathbf e_i,$ $\|\beta\|=(\sum_{i=0}^k\beta_i^2)^{1/2}.$ 
In the context of derivative based regularization, we need the definition of symmetric tensors (representing derivatives) defined at any point in the space.  For this, we define the concept of tensor fields. Since, we are dealing with images, we will call them as tensor images as they define a tensor at each location in the 2D space.
\begin{definition}{(Continuous Tensor Images (fields))}
	A continuous k-tensor image (field) $ \alpha: \mathbb R^2 \rightarrow \mathbb \mathbb C^{(k)}$ assigns a tensor at each point of the two dimensional space. For e.g. a function $f:\mathbb R^2\rightarrow \mathbb R$ is a 0-dimensional tensor field. Similarly, instead of the complete 2-D space, we can define tensor fields confined to an open subset $\Omega \subset \mathbb R^2$ as follows:\\
	A continuous k-tensor image (field) $ \alpha: \Omega \rightarrow \mathbb \mathbb C^{(k)}$ assigns a tensor at each point of $\Omega$. In order to define the discrete total generalized variation,   it is required that a tensor is defined at each pixel location. Analogously, we can also define symmetric k-tensor fields. We denote the set of continuous symmetric k-tensor fields on $\Omega$ as $\mathbb S_{\Omega}^{(k)}.$
\end{definition}
\begin{definition}(Discrete Tensor Images)  A discrete k-tensor image (of $N$ pixels) $\mathcal \alpha: \{1,2,...,N\} \rightarrow \mathbb \mathbb C^{(k)}$ assigns each pixel with a k-tensor. Similarly, we can define discrete symmetric tensor images. We denote the set of discrete symmetric k-tensor fields on $N$  ordered pixel locations as $\mathbb S_{N}^{(k)}.$

\end{definition}
\begin{definition} A discrete symmetric k-tensor image (of $N$ pixels) $\mathcal \alpha: \{1,2,...,N\} \rightarrow \mathbb \mathbb S^{(k)}$ assigns each pixel with a \textbf{symmetric} k-tensor. We denote the set of discrete symmetric k-tensor fields on $N$  ordered pixel locations as $\mathbb S_{N}^{(k)}.$

\end{definition}
\begin{definition}(Symmetric derivative ($\mathcal E^{(k)}$))
	As TGV involves iterated derivative, we need to define the  symmetric derivative $\mathcal E ^{(k)}:\mathbb S^{(k)}_{\Omega}\rightarrow \mathbb S^{(k+1)}_{\Omega}.$ Consider any pixel location $(p,q)\text {in the set } \Omega.$  For this, let $\pmb \eta\in \mathbb S_{\Omega}^{(k)}$ such that $ \pmb \eta(p,q)=\sum_{\bn (p,q) \in \Omega}\eta_{\bn i}(p,q)\pmb \omega_{\bn i}$. The symmetric iterated derivative in continuous domain is defined as: 
	$$(\mathcal E^{(k)}(\pmb \eta))(p,q)=|||^{(k+1)}\Big[\sum_{\bn i \in \{0,1\}^k}\Big(({\frac{  \partial\eta_{\bn i}}{\partial x} })(p,q) \pmb \omega_{\bn i}\otimes \omega_0+({\frac{  \partial\eta_{\bn i}}{\partial y} })(p,q)\pmb \omega_{\bn i}\otimes \omega_1\Big)\Big].$$
	In the above equation, $\eta_{\bn i}:\mathbb R^2\rightarrow \mathbb R$ are real valued functions. Therefore, the partial derivatives $\frac {\partial}{\partial x}$ and $\frac {\partial}{\partial y}$ are clearly defined.
\end{definition}
\begin{definition}\label{ch4:defn:epsilon}(Symmetric derivative in discrete form $(\epsilon_k)$)
	In order to define the symmetrized derivative in   discrete form,  we first consider discrete tensor image (of $N$ pixels) $\pmb \eta\in \mathbb S_N^{(k)}$ such that $ \pmb \eta(j)=\sum_{\bn i \in \{0,1\}^k}\eta_{\bn i}(j)\pmb \omega_{i}$ for $j\in \{1,2,..N\}$, we replace $\partial /\partial x$ with $\bn D_x$, and similarly $\partial/\partial y$ with $\bn D_y$. As in continous setting, $\eta_{\bn i}:\{1,2,...,N\}\rightarrow \mathbb R$ are real valued functions (denoting grayscale images.) Therefore, $\bn D_x$ and $\bn D_y$ are clearly defined, and $\bn D_x\eta_{\bn i}$ and $\bn D_y\eta_{\bn i}$ are real valued derivative images. 
	\begin{align}
	&( \epsilon_{k}(\pmb \eta))(j)\df|||^{(k+1)}\Big[\sum_{\bn i \in \{0,1\}^k}\Big(({\bn D_x\eta_{\bn i}} )(j)\pmb \omega_{\bn i}\otimes \omega_0+({{  \bn D_y\eta_{\bn i}} } ) (j)\pmb \omega_{\bn i}\otimes  \omega_1\Big)\Big].
	\end{align}
\end{definition}
In the above definition, $\epsilon_k$ was defined on the standard basis. We extend the above definition to be defined on  the orthongonal basis for the symmetric tensors.
\begin{lemma}\label{ch4:remark:eq}
	Consider any $\pmb \beta\in \mathbb  S_N^{(k)}$ given as $\pmb \beta(i)=\sum_{r=0}^k \beta_r (i)\bn e_r^{(k)}$ for any pixel location $i$. The operation of symmetric derivative on $\beta$ can be written as:
	\begin{equation}
	\mathbb (\epsilon_k(\pmb \beta))(i)=|||^{(k+1)}\Big[\sum_{r=0}^{(k)}\Big(({\mathbf D_x  \beta_r })(i)\mathbf e_{r}^{(k)}\otimes \omega_0+({\mathbf D_y  \beta_r })(i)\mathbf e_{r}^{(k)}\otimes  \omega_1\Big)\Big].\end{equation} 
	
\end{lemma}
\begin{proof}
	We prove the above result from the  \cref{ch4:defn:epsilon}. Consider any symmetric tensor image $\pmb \eta(r)=\sum_{\bn i\in\{0,1\}^N}\eta_{\bn i}(r)\pmb \omega_{\bn i}\in \mathbb S_N^{(k)}.$ Since this tensor image is symmetric we can write $\pmb \eta (r)$  as a linear combination of the orthogonal basis vectors of $\mathbb S^{(k)}$ (the space of symmetric tensors) as $\pmb \eta(r)=\sum_{j=0}\eta'_j(r)\bn e^{(k)}_j$ for each pixel location $r$. Now we can divide $\{0,1\}^k$ into disjoint sets $T_j$ where  $T_j\df \Set{\bn i\in \{0,1\}^k}{s(\bn i)=j}.$ Also, it can be observed that $\{0,1\}^k=\cup_{j=0}^k T_j.$ With this the symmetric tensor $\pmb \eta=\sum_{r=0}^k\sum_{\bn i \in T_j} \eta_{\bn i}\pmb \omega_{\bn i}.$ As $\pmb \eta$ is symmetric,  $\eta_i$ remains same over the set $T_r$ for any $r$, i.e. if $\bn i^{(1)}$ and $\bn i^{(2)}$ both belong to $T_r$ then $\eta_{\bn i^{(1)}}=\eta_{\bn i^{(2)}}.$ With this we can define, $\eta_j=\eta_{\bn i}$ for each $\bn i\in T_j.$ Now, we can write $\pmb \eta(r)=\sum_{r=0}^k\eta_j(r)\sum_{\bn i \in T_r} \pmb \omega_{\bn i}.$ Using the definition of $\bn e_j^{(k)}$ we can conclude that: $\eta_j(r)\sqrt{\Mycomb[k]{j}}=\eta_j'(r)$ for any r. Now, we invoke the definition of symmetric derivative:
	\begin{align}
	( \epsilon_{k}(\pmb \eta))(r)&=|||^{(k+1)}\Big[\sum_{\bn i \in \{0,1\}^k}\Big(({\bn D_x\eta_{\bn i}} )(r)\pmb \omega_{\bn i}\otimes \omega_0+({{  \bn D_y\eta_{\bn i}} } ) (r)\pmb \omega_{\bn i}\otimes  \omega_1\Big)\Big]\\&=|||^{(k+1)}\Big[\sum_{j=0}^k\sum_{\bn i \in T_j}\Big(({\bn D_x\eta_{\bn i}} )(r)\pmb \omega_{\bn i}\otimes \omega_0+({{  \bn D_y\eta_{\bn i}} } ) (r)\pmb \omega_{\bn i}\otimes  \omega_1\Big)\Big]\\&=|||^{(k+1)}\Big[\sum_{j=0}^k\Big(({\bn D_x\eta_{j}} )(r)\sum_{\bn i \in T_j}\pmb \omega_{\bn i}\otimes \omega_0+({{  \bn D_y\eta_{j}} } ) (r)\sum_{\bn i \in T_j}\pmb \omega_{\bn i}\otimes  \omega_1\Big)\Big]\\&=|||^{(k+1)}\Big[\sum_{j=0}^k\Big(({\bn D_x\eta'_{j}} )(r)\bn e_j^{(k)}\otimes \omega_0+({{  \bn D_y\eta'_{j}} } ) (r)\bn e_j^{(k)}\otimes  \omega_1\Big)\Big]
	\end{align}
\end{proof}
\begin{definition}\label{def:cont_tgv}The Total Generalized Variation (TGV) in continuous form as given by \cite{tgv} can be written as:
	\begin{equation}\label{eq:tgv_discrete}
	\mathcal {TGV}^n(\mathbf g)=\inf_{u_i\in \mathbb S^{(i)}_{ \Omega},u_0=(\mathbf g),u_{n}=0}\sum_{i=0}^{n-1}\alpha_{n-i-1}\|\mathcal E^{(i)}u_i-u_{i+1}\|.
	\end{equation}
\end{definition}
\begin{definition}\label{def:discrete_tgv}The above definition of Total Generalized Variation (TGV) can be discretized by replacing $\mathcal E^{(i)}$ with its discrete counterpart $\epsilon_i$ as:
	\begin{equation}\label{eq:tgv_orignal}
	{TGV}^n(\mathbf g)=\inf_{\bn u_i\in \mathbb S^{(i)}_{N},\bn u_0=(\mathbf g),\bn u_{n}=\bn 0}\sum_{i=0}^{n-1}\alpha_{n-i-1}\|\mathcal \epsilon_i \bn u_i-\bn u_{i+1}\|.
	\end{equation}
\end{definition}

\subsection{Some results on tensors used for deriving TGV formulations}
The following result relates the two basis defined in the previous section.
\begin{proposition}
	
	\label{cor:rel}(Relation between standard basis of $\mathbb C^{(k)}$ and orthonormal basis of $\mathbb S^{(k)}$)
	Consider $\alpha_j=\pmb \omega_{\bn p}=\omega_{p_1}\otimes...\otimes\omega_{p_k} \in \mathbb C^{(k)},$ such that $s(\bn p)=\sum_{l=1}^kp_l=j$ for some $j \in \{0,1,...,k\}$. Then, $ \ \ |||^{(k)}{(\alpha_j)}= \frac{1}{\sqrt{\Mycomb[k]{j}}}\ \ \mathbf{e}_j^{(k)}.$
\end{proposition}
Proof: By \cref{cor:e} we have,
$$\mathbf e_j^{(k)}=\frac{1}{\sqrt{\Mycomb[k]{j}}}\sum_{\mathbf m \in \{0,1\}^{k} , s(\bn m)=j} \pmb \omega_{\bn m} .$$
Applying the (linear) operator $|||^{(k)}$ on both sides we get,

$$|||^{(k)}\mathbf e_j^{(k)}=\frac{1}{\sqrt{\Mycomb[k]{j}}}\sum_{\mathbf m \in \{0,1\}^{k} ,s(\bn m)=j}|||^{(k)}\pmb \omega_{\bn m} .$$ Now, from \cref{def:sym} it can be seen that all elements inside the summation are equal. Therefore, we get,
$$|||^{(k)}\mathbf e_j^{(k)}=\frac{1}{\sqrt{\Mycomb[k]{j}}} {\Mycomb[k]{j}}|||^{(k)}\pmb \omega_{\bn m} .$$ Using the fact that $|||^{(k)}$ is the projection and $\mathbf e_{j}^{(k)}$ is a symmetric tensor gives the result.
\\The following proposition identifies tensor images with matrices which allows us to represent TGV in a tensor-free form.
\begin{proposition}
	\label{ch4:remark_iso}
	The set of discrete tensor images $S_N^{(k)}$ is a vector space (over $\mathbb R$) of dimension $N\times (k+1)$ and therefore, isomorphic to $\mathbb R^{N\times (k+1)}.$
\end{proposition} 
\begin{proof}To see that $\mathbb S_N^{(k)}$ is a vector space, one can verify that the linear combination of any two elements of  $\mathbb S_N^{(k)}$ is in $\mathbb S_N^{(k)}$. \\
	To show that $\mathbb S_N^{(k)}$ is isomorphic to $\mathbb R^{N\times(k+1)}$, we show that the dimension of $\mathbb S_N^{(k)}$ is $N\times(k+1).$To this end we prove that the basis of $\mathbb S_N^{(k)}$ is the set of tensor images $\{f_{i,j}:\{1,2,...,N\}\rightarrow \mathbb S^{(k)}|i\in\{1,2,...,N\},j\in \{0,1,...,k\}\}$ (of $N$ pixels) defined for each pixel $r$ as:
	$f_{i,j}(r)=\bn e^{(j)}_k$ if $r=i$ and $f_{i,j}(r)=0$ if $r\neq i.$ To show that the set spans, consider any $\pmb \beta\in \mathbb S_{N}^{(k)}$ given as $\pmb \beta(r)=\sum_{j=0}^k\beta_j(r)\bn e^{(k)}_j$ in the orthogonal basis  for the space of symmetric k-tensors. Now, for any $r$, $$\pmb \beta(r)=\sum_{j=0}^k\beta_j(r)\bn e^{(k)}_j.$$ By definition of $f_{i,j}s$ we can write:
	\begin{align}
	\pmb \beta(r)=&\sum_{j=0}^k\beta_j(r)f_{r,j}(r)\\=&\big[\sum_{i=1}^N\sum_{j=0}^k\beta_j(i)f_{i,j}\big](r).
	\end{align}
	Therefore, any element in $\mathbb S_N^{(k)}$ can be written as a linear combination of $f_{i,j}$s. To show that they are linearly independent we consider the linear combination $\sum_{i=1}^N\sum_{j=0}^{k}a_{i,j}f_{i,j}=\bn 0.$ Choose any $r \in \{1,2,...,N\}.$ For this $r$ we have: $ \sum_{i=1}^N\sum_{j=0}^{k}a_{i,j}f_{i,j}(r)=\bn 0$. By the definition of $f_{i,j},$ $\sum_{j=0}^{k}a_{r,j}f_{r,j}(r)=\bn 0.$ This means,$\sum_{j=0}^{k}a_{r,j}\bn e^{(k)}_j=\bn 0$. Since, $\bn e^{(k)}_j$s are linearly independent, $a_{r,j}=0$ for $j=0,1,..,k.$ As $r$ was arbitrarily chosen, $a_{r,j}=0$ for all  $r\in \{1,2,...N\}$ and $j\in \{ 0,...,k\}.$
\end{proof}

\begin{remark}(Isomorphism between tensor images and matrices)\label{ch4:remark:isom}
	As a result of the above theorem, any discrete symmetric tensor field $\pmb \alpha \in \mathbb S^{(k)}_N$ can be represented by a matrix of size $N\times(k+1).$ Now, we explicitly define the isomorphism between the two vector spaces. Consider any symmetric tensor field $\pmb \alpha$, defined for any pixel index $r\in \{1,2,...,N\}$ as $\pmb \alpha(r)= \sum_{j=0}^k\alpha_j(r)(\bn e_j^{(k)})$ .Here, $\alpha_j(r)$ denotes the  coefficient for the basis $\bn e_j^{(k)}.$   The isomorphism 
	$\phi_k: \mathbb S^{(k)}_N\rightarrow {\mathbb R^{N\times(k+1)}}$ is given as: $\phi_k(\pmb \alpha)=\begin{bmatrix}\alpha_0(1)&...&\alpha_k (1)\\
	\alpha_0(2)&...&\alpha_k (2)\\
	\vdots &\vdots&\vdots&\\
	\alpha_0(N)&...&\alpha_k(N)\end{bmatrix}.$
\end{remark}
\begin{remark}
	The set of tensor images $\mathbb C_N^{(k)}$ is a vector space (over $\mathbb R$) of dimension ${N\times 2^k}.$ Therefore, it is isomorphic to $\mathbb R^{N\times 2^k}$, and  we denote the isomorphism as $\psi_k:\mathbb C^{(k)}_N\rightarrow \mathbb R^{N\times 2^k}.$
\end{remark}
The proof of the above remark is similar to the proof of \cref{ch4:remark:isom}. Hence, we skip the proof here.
\begin{definition}
	The mixed norm ($\|\cdot\|_{1,2}$) for any $\mathbf M \in \mathbb R^{N\times p} $ is defined as:$$\|\mathbf M\|_{1,2}=\sum_{j=1}^{N}(\sum_{i=1}^{p}\mathbf M_{j,i}^2)^{1/2}.$$
\end{definition}
\begin{remark}\label{ch4:lemma:norm}
	For any $\pmb \alpha \in \mathbb S^{(k)}_n, \|\alpha\|\df\sum_{i=1}^n\|\alpha(i)\|$. With this definition we get,$\|\pmb \alpha\|=\|\phi_k(\pmb \alpha)\|_{1,2}.$
\end{remark}
\begin{proof}
	To prove this we need to show that for any $\beta=\sum_{r=0}^{k}\beta_r\bn e^{(k)}_r\in \mathbb S^{(k)},\ \ \|\mathbb \beta\|=\big(\sum_{r=0}^{k}\beta_r^2\big)^{\frac{1}{2}}.$ This follows from the fact that $\langle \bn e^{(k)}_j,\bn e^{(k)}_l\rangle=1$ if $j=l$ and $0$ else.
\end{proof}
\subsection{Organization of Proofs}

To establish the results regarding the representation of Total Generalized Variation (TGV), we start by considering the discrete definition of TGV (see \cref{def:discrete_tgv}), which corresponds to a discretized rendition of the definition presented in \cite{tgv}. Through the utilization of \cref{prop:direct} (that gives the iterated gradient in its matrix form), we substantiate the theorem that presents the direct representation of TGV (refer to \cref{ch4:th:tgvdir}). Furthermore, starting with the same definition, we harness the insights provided by \cref{ch4:remark:eq} to facilitate the expression from \cref{def:discrete_tgv} in terms of the basis of symmetric tensors. This effort culminates in proving the theorem that formulates the compact representation (see \cref{ch4:th:rep}). Importantly, \cref{ch4:remark:dirrep} independently demonstrates the equivalence of both of these forms. Hence, an alternative approach to establishing the compact form involves proving the direct form and using \cref{ch4:remark:dirrep}.\subsection{Proofs of the theorem for representation of Total Generalized Variation}
As we need to give a tensor free representation, we give a matrix equivalent of $\epsilon_k$ (symmetric derivative operator) with the help of the following proposition. 
\begin{proposition} \label{prop:direct}For any tensor image $\pmb \eta \in \mathbb S_N^{(k)},$ we have $\psi_{k+1} \epsilon_k\pmb\eta=\Pi^{(k+1)}\mathcal D_{2^k-1}\psi_k(\pmb\eta).$
\end{proposition}
\begin{proof}
	We will prove this by proving the following for each index $(r,l)$:  $$(\psi_{k+1} \epsilon_k\pmb\eta)(r,l)=(\Pi^{(k+1)}\mathcal D_{2^k-1}\psi_k(\pmb\eta))(r,l).$$  We start with \cref{ch4:defn:epsilon}, the definition of symmetric gradient $\epsilon$. For any pixel location $r\in\{1,..,N\}$ we can write $\pmb  \eta$ as a linear combination of the basis vectors $\pmb \omega_{\bn i}$ with coefficients $\eta_{\bn i}$ as: 
	$$\epsilon_k\pmb \eta(r)\df|||^{(k+1)}\Big[\sum_{\bn i \in \{0,1\}^k}\Big(({\bn D_x\eta_{\bn i}} )(r)\pmb \omega_{\bn i}\otimes \omega_0+({{  \bn D_y\eta_{\bn i}} } ) (r)\pmb \omega_{\bn i}\otimes  \omega_1\Big)\Big].$$
	Let $\bn w_{\bn j}$ be the tuple $(\bn w_{j_1},...,\bn w_{j_k})$ (recall the definitions $\bn w_0=(0,1)^T$ and $\bn w_1=(1,0)^T$) for some $\bn j=(j_1,...,j_k) \in\{0,1\}^k$. Also, observe that the $(r,l)$ element of $\psi_{k+1}\epsilon_k\pmb \eta$ is $(\epsilon_k(\pmb \eta)(r))(\bn w_{\bn j})$, where $\bn j =b(l).$ With this we have, $$\big(\epsilon_k\pmb \eta(r)\big)(\bn w_{\bn j})=\Big[\sum_{\bn i \in \{0,1\}^k}|||^{(k+1)}\Big(({\bn D_x\eta_{\bn i}} )(r)(\pmb \omega_{\bn i}\otimes\omega_0)(\bn w_{\bn j}) +({{  \bn D_y\eta_{\bn i}} } ) (r) (\pmb\omega_{\bn i}\otimes\omega_1)(\bn w_{\bn j})\Big)\Big].$$
	Let $B$ be the matrix $\mathcal D_{2^k-1}\psi_k(\pmb \eta)$, then  by the definition of $\mathcal D_{2^k-1}$ we have that $B_{r,b^{-1}((\bn i,0))}=({{  \bn D_x\eta_{\bn i}} } ) (r),$ similarly, $B_{r,b^{-1}((\bn i,1))}=({{  \bn D_y\eta_{\bn i}} } ) (r).$ This is because at odd indices we have $\bn D_y$ and at even indices we have $\bn D_x$. With these we have: \\
	$$\big(\epsilon_k\pmb \eta(r)\big)(\bn w_{\bn j})=\frac{1}{k+1!}\Big[\sum_{\bn i \in \{0,1\}^k}\sum_{\pi \in S_{k+1}}\Big(B_{r,b^{-1}((\bn i,0))}(\pmb \omega_{\bn i}\otimes\omega_0)(\bn w_{ f_{\pi}(\bn j)}) +B_{r,b^{-1}((\bn i,1))} (\pmb\omega_{\bn i}\otimes\omega_1)(\bn w_{f_{\pi}(\bn j)} )\Big)\Big].$$ 
	The two terms inside the summation can be combined into a single summation by defining a bigger vector $\bn m =(\bn i,*)$. Here, $*$ can be $0$ or 1. Now, the expression becomes:
	$$\big(\epsilon_k\pmb \eta(r)\big)(\bn w_{\bn j})=\frac{1}{k+1!}\Big[\sum_{\bn m\in \{0,1\}^{k+1}}\sum_{\pi \in S_{k+1}}\Big(B_{r,b^{-1}(\bn m)}(\pmb \omega_{\bn m})(\bn w_{ f_{\pi}(\bn j)})\Big)\Big].$$ 
	
	Now, $(\pmb \omega_{\bn m})(\bn w_{f_{\pi}(\bn j)} )$ is one if and only if $\bn m=f_{\pi}(\bn j)$ and $0$ else. Observe that $\bn j=b(l)$.
	With this, we have:
	$$\big(\epsilon_k\pmb \eta(r)\big)(\bn w_{\bn j})=\frac{1}{k+1!}\Big[\sum_{\pi \in S_{k+1}}\Big(B_{r,b^{-1}(f_{\pi}(b(l)))}\Big)\Big]=(\Pi^{(k+1)}B)_{r,l}=(\Pi^{(k+1)} \mathcal D_{2^k-1}(\psi_k)(\pmb \eta))_{r,l}.$$

	
\end{proof}
Now, with the help of the above proposition we give a tensor free representation of TGV.
\tgvdir*
\begin{proof}
	Recalling the definition of the   total generalized variation  (\cref{def:discrete_tgv}):
	\begin{equation}
	{TGV}^n(\mathbf g)=\inf_{\bn u_i\in \mathbb S^{(i)}_{N},\bn u_0=(\mathbf g),\bn u_{n}=\bn 0}\sum_{i=0}^{n-1}\alpha_{n-i-1}\|\mathcal \epsilon_i \bn u_i-\bn u_{i+1}\|.
	\end{equation}
	Let $\pmb \eta(j)=\sum_{\bn p \in \{0,1\}^i}\eta_{\bn p}(j)\pmb \omega_{\bn p}$ be any symmetric i-tensor image, 
	then from \cref{prop:direct} we obtain for any pixel location $j$: 
	\begin{align}
	&( \epsilon_{i}(\pmb \eta))(j)&\\&=|||^{(i+1)}\Big[\sum_{\bn p \in \{0,1\}^i}\Big(({\bn D_x\eta_{\bn p}} )(j)\pmb \omega_{\bn p}\otimes \omega_0+({{  \bn D_y\eta_{\bn p}} } ) (j)\pmb \omega_{\bn p}\otimes  \omega_1\Big)\Big].\\&
	=(\psi_{i+1}^{-1}\Pi^{(i+1)}\mathcal D_{2^i-1}\psi_i(\pmb \eta))(j)
	\end{align}
	With the above result $TGV^n$ becomes:
	\begin{align}
	{TGV}^n(\mathbf g)&=\inf_{\bn u_i\in \mathbb S^{(i)}_{N},\bn u_0=(\mathbf g),\bn u_{n}=\bn 0}\sum_{i=0}^{n-1}\alpha_{n-i-1}\|\psi_{i+1}^{-1}\Pi^{(i+1)}\mathcal D_{2^i-1}\psi_i \bn  u_i - \bn u_{i+1}\|\\&\stackrel{(p)}{=}\inf_{\bn u_i\in \mathbb S^{(i)}_{N},\bn u_0=(\mathbf g),\bn u_{n}\bn 0}\sum_{i=0}^{n-1}\alpha_{n-i-1}\|\psi_{i+1}^{-1}\Pi^{(i+1)}\mathcal D_{2^i-1}\psi_i \bn  u_i - \psi_{i+1}^{-1}\psi_{i+1}\bn u_{i+1}\|.\\&\stackrel{(q)}{=}\inf_{\bn v_i\in  \mathcal R(\Pi^{(i)}),\bn v_0=(\mathbf g),\bn u_{n}=\bn 0}\sum_{i=0}^{n-1}\alpha_{n-i-1}\|\Pi^{(i+1)}\mathcal D_{2^i-1} \bn  v_i - \bn v_{i+1}\|_{1,2}.
	\end{align}
	Here, (p) follows from the definition of the norm of the tensor and (q) follows by substituting $\psi_i\bn u_i=\bn v_i$ for $i=0,1,..,n$.
\end{proof}
To give the compact and tensor free represntation we give the matrix equivalent of $\epsilon_k$ using the orthogonal basis of symmetric tensor images rather than the standard basis of tensor images. This greatly simplifies the TGV expression.
\begin{proposition}\label{prop:eps}(Derivation of symmetric derivative operator $\epsilon_k$  in compact matrix form using  basis $\{\mathbf e_j^{(k)}\}_{j=0}^k$ of $\mathbb S_N^{(k)}$ ) 
	
	The linear operator $\phi_{k+1}\epsilon_k\phi_k^{-1}: \mathbb R^{N\times (k+1)}\rightarrow \mathbb R^{N\times (k+2)}$ can be written as a composition of two linear operators. $\mathcal A_k:\mathbb R^{N \times (2k+2) }\rightarrow \mathbb R^{N \times (k+2) }$ and $\mathcal D_k:\mathbb R^{N \times (k+1) }\rightarrow \mathbb R^{N \times (2k+2) }.$ Here, $\mathcal A_k(\mathbf z)=\mathbf z\cdot M_k$, $M_k=\begin{bmatrix} &1 & 0 &0&...&0&0\\&0&\frac{\sqrt{\Mycomb[k]{1}}}{\sqrt {\Mycomb[k+1]{1}}}&0&...&0&0\\ &0&\frac{\sqrt{\Mycomb[k]{0}}}{\sqrt {\Mycomb[k+1]{1}}}&0&...&0&0\\&0&0&\frac{\sqrt{\Mycomb[k]{2}}}{\sqrt {\Mycomb[k+1]{2}}}&...&0&0\\
	&0&0&\frac{\sqrt{\Mycomb[k]{1}}}{\sqrt {\Mycomb[k+1]{2}}}&...&0&0\\
	&\vdots&\vdots&\vdots&...&\vdots&\vdots\\
	&0&0&0&...&\frac{\sqrt{\Mycomb[k]{k}}}{\sqrt {\Mycomb[k+1]{k}}}&0\\
	&0&0&0&...&\frac{\sqrt{\Mycomb[k]{k-1}}}{\sqrt {\Mycomb[k+1]{k}}}&0\\
	&0&0&0&...&0&1\\
	\end{bmatrix};$ $$\mathcal D_k([\mathbf y_0, \mathbf y_1,...,\mathbf y_k])=[\mathbf D_x \mathbf y_0,\mathbf D_y \mathbf y_0,\mathbf D_x \mathbf y_1,\mathbf D_y \mathbf y_1,...,\mathbf D_x \mathbf y_k,\mathbf D_y \mathbf y_k].$$Also, $\mathcal A_k \mathcal A_k^T=\mathcal I.$ \label{matrix}\end{proposition}
\begin{proof}
	Let $\pmb \beta=\sum_{r=0}^{k}\beta_r \mathbf e_{r}^{(k)}\in \mathbb S^{(k)}_N.$ Consider any pixel location $i$, the operator $\epsilon_k$ is given as:\begin{equation}
	\mathbb (\epsilon_k(\pmb \beta))(i)=|||^{(k+1)}\Big[\sum_{r=0}^{(k)}\Big(({\mathbf D_x  \beta_r })(i)\mathbf e_{r}^{(k)}\otimes \omega_0+({\mathbf D_y  \beta_r })(i)\mathbf e_{r}^{(k)}\otimes  \omega_1\Big)\Big].\end{equation}
	Since, $|||^{(k+1)}$ is a linear operator we have,
	\begin{equation}
	\mathbb (\epsilon_k(\pmb \beta))(i)=\sum_{r=0}^{(k)}\Big(({\mathbf D_x  \beta_r })(i)|||^{(k+1)}(\mathbf e_{r}^{(k)}\otimes \omega_0)+({\mathbf D_y  \beta_r })(i)\mathbf |||^{(k+1)}(\mathbf e_{r}^{(k)}\otimes  \omega_1)\Big).\end{equation} Consider,
	\begin{eqnarray} \label{eq:discrete1} \nonumber |||^{(k+1)}(\mathbf e_{r}^{(k)}\otimes\omega_0)\stackrel{\text{\cref{cor:e}}}{=}&\frac{1}{\sqrt{\Mycomb[k]{r}}}\sum_{\mathbf p\in\{0,1\}^k,\sum_l p_l=r} |||^{(k+1)}\omega_{p_1}\otimes...\otimes\omega_{p_k}\otimes\omega_0\\ &\stackrel{\text{\cref{cor:rel}}}{=}\frac{1}{\sqrt{\Mycomb[k]{r}}}\Mycomb[k]{r}\frac{\mathbf e^{(k+1)}_r}{\sqrt{\Mycomb[k+1]{r}}}.\end{eqnarray} Similarly,
	\begin{equation} \label{eq:discrete2}|||^{(k+1)}(\mathbf e_{r}^{(k)}\otimes\omega_1)=\frac{\sqrt{\Mycomb[k]{r}}}{\sqrt{\Mycomb[k+1]{r+1}}}{\mathbf e^{(k+1)}_{r+1}}.\end{equation}
	Using eq. \eqref{eq:discrete1} and \cref{eq:discrete2} we get,
	\begin{eqnarray*}
		&\mathbb (\epsilon_k(\pmb \beta))(i)=\sum_{r=0}^{k}\Big(({\mathbf D_x  \beta_r })(i)\frac{\sqrt{\Mycomb[k]{r}}}{\sqrt{\Mycomb[k+1]{r}}}{\mathbf e^{(k+1)}_{r}}+({\mathbf D_y  \beta_r })(i) \frac{\sqrt{\Mycomb[k]{r}}}{\sqrt{\Mycomb[k+1]{r+1}}}{\mathbf e^{(k+1)}_{r+1}}\Big)\\&=({\mathbf D_x  \beta_0 })(i)\mathbf e^{(k+1)}_0+({\mathbf D_y  \beta_k })(i)\mathbf e_{k+1}^{(k+1)}+\sum_{r=1}^{(k)}\Big(({\mathbf D_x  \beta_r })(i)\frac{\sqrt{\Mycomb[k]{r}}}{\sqrt{\Mycomb[k+1]{r}}}+({\mathbf D_y  \beta_{r-1} })(i) \frac{\sqrt{\Mycomb[k]{r-1}}}{\sqrt{\Mycomb[k+1]{r}}}\Big){\mathbf e^{(k+1)}_{r}}.\end{eqnarray*}
	
	By collecting the coefficients of $\bn e_{j}^{(k+1)}$ for $j=0,1,...,k+1$, it can be seen that $\phi_{k+1}\epsilon_k\phi_k^{-1}=\mathcal A_k\circ \mathcal D_k.$ To show that $\mathcal A_k\circ\mathcal A_k^T=\mathcal I:$ first we see that $\mathcal A_k^T(\mathbf z)=\mathbf z\cdot M_k^T.$ Therefore, it is sufficient show that $M_k^TM_k=\mathbf I$. On computation, it can be seen that $M_k^T M_k$ is  a diagonal matrix with values $\frac{\Mycomb[k]{r}}{\Mycomb[k+1]{r}}+\frac{\Mycomb[k]{r-1}}{\Mycomb[k+1]{r}}.$ Since, $\Mycomb[k]{r-1}+\Mycomb[k]{r}=\Mycomb[k+1]{r},$ each diagonal element is $1.$ Hence, $M_k^TM_k=\mathbf I.$

\end{proof}
Now, we finally give the compact representation of TGV using the following theorem.
\tgvd*
\begin{proof}
	We begin the proof from  the definition of the   total generalized variation  (\cref{def:discrete_tgv}):
	\begin{equation}
	{TGV}^n(\mathbf g)=\inf_{\bn u_i\in \mathbb S^{(i)}_{N},\bn u_0=(\mathbf g),\bn u_{n}=\bn 0}\sum_{i=0}^{n-1}\alpha_{n-i-1}\|\mathcal \epsilon_i \bn u_i-\bn u_{i+1}\|.
	\end{equation}
	
	With the above proposition \cref{prop:eps}, we have:
	\begin{align}TGV^n(\mathbf g)&=\inf_{\bn u_i\in \mathbb S^{(i)}_{ N},\bn u_0=\mathbf g,\bn u_{n}=0}\sum_{i=0}^{n-1}\alpha_{n-i-1}\|\phi_{i+1}^{-1}\mathcal A_i\circ \mathcal D_i\phi_i \bn u_i-\bn u_{i+1}\|\\&
	=\inf_{\bn u_i\in \mathbb S^{(i)}_{ N},u_0=(\mathbf g),u_{n}=0}\sum_{i=0}^{n-1}\alpha_{n-i-1}\|\phi_{i+1}^{-1}\mathcal A_i\circ \mathcal D_i\phi_i \bn u_i-\phi_{i+1}^{-1}\phi_{i+1}\bn u_{i+1}\|\\&\stackrel{(r)}{=}\inf_{\bn p_i\in \mathbb R^{N\times (i+1) },\bn p_0=\mathbf g,\bn p_{n}=0}\sum_{i=0}^{n-1}\alpha_{n-i-1}\|\mathcal A_i\circ \mathcal D_i \bn p_i-\bn p_{i+1}\|_{1,2}\\
	\end{align}
	In the above expression (r) follows from \cref{ch4:lemma:norm} and substituting $\phi_i \bn u_i=\bn p_i$ for $i=0,1,..,n$ we get the result.
\end{proof}
We use the following lemma for proving the equivalence of the two given reprsentations.
\begin{lemma}\label{ch4:lemma:eq_cost}
	Consider any $\bn p=[\bn p_0,...,\bn p_k]\in \mathbb R^{N\times(k+1)}$, $\bn q=[\bn q_0,...,\bn q_{k+1}]\in \mathbb R^{N\times(k+2)}$, $\bn u=[\bn u_0,...,\bn u_{2^{k}-1}]\in \mathbb R^{N\times 2^k}$, $\bn v=[\bn v_0,...,\bn v_{2^{k+1}-1}]\in \mathbb R^{N\times(2^{k+1})}$. If $\bn u_j=\frac{\bn p_{s(b(j))}}{\sqrt{\Mycomb[k]{s(b(j))}}}$ for $j=0,...,2^k-1$ and $\bn v_l=\frac{\bn p_{s(b(l))}}{\sqrt{\Mycomb[k]{s(b(l))}}}$ for $l=0,...,2^{k+1}-1$ then $\|\mathcal A_k\mathcal D_k\bn p-\bn q\|_{1,2}=\|\Pi^{(k+1)}\mathcal D_{2^k-1}\bn u-\bn v\|_{1,2}.$
\end{lemma}
\begin{proof}
	By the given definition of $\bn u$ we have:
	\begin{align}\mathcal D_{2^{k}-1}\bn u&=[{\bn D_x \bn u_0},{\bn D_y \bn u_0},{\bn D_x \bn u_1},{\bn D_y \bn u_1},...,{\bn D_x \bn u_k},{\bn D_y \bn p_k}]\\&=[\frac{\bn D_x \bn p_0}{\sqrt{\Mycomb[k]{0}}},\frac{\bn D_y \bn p_0}{\sqrt{\Mycomb[k]{0}}},\frac{\bn D_x \bn p_1}{\sqrt{\Mycomb[k]{1}}},\frac{\bn D_y \bn p_1}{\sqrt{\Mycomb[k]{1}}},...,\frac{\bn D_x \bn p_k}{\sqrt{\Mycomb[k]{k}}},\frac{\bn D_y \bn p_k}{\sqrt{\Mycomb[k]{k}}}].\end{align}
	For any $i\in \mathbb N$, define the set $[i]\df \Set{r\in \mathbb N\cup{0}}{r<i}$. Now, the definition of $\mathcal D_{2^k-1}(\cdot)$ implies  that for any even $r\in [2^{k+1}]$, $(\mathcal D_{2^{k}-1}\bn u)_{:,r}=\bn D_x \bn u_{r/2}$, and for $r$ odd $(\mathcal D_{2^k-1}\bn u)_{:,r}=\bn D_y \bn u_{(r-1)/2}$. As multiplication by 2 only shifts the binary code to one place left, $s(b(r/2))=s(b(r))$  (for $r$ even) $s(b((r'-1)/2))=s(b(r'))-1.$ (for $r'$ odd). Further from the given relation between $\bn u$ and $\bn p$ we have for \begin{itemize}\item $r$ even: $(\mathcal D_{2^k-1}\bn u)_{:,r}=\frac{1}{\sqrt{\Mycomb[k]{s(b(r))}}}\bn D_x \bn p_{s(b(r))}$
		\item $r$ odd:  $(\mathcal D_{2^k-1}\bn u)_{:,r}=\frac{1}{\sqrt{\Mycomb[k]{s(b(r))-1}}}\bn D_y \bn p_{s(b(r))-1}$\end{itemize}
	As $\Pi^{(k+1)}$ only sums the columns, we compute the $J^{th}$ column of $\Pi^{(k+1)}\mathcal D_{2^{k}-1}\bn u$, $(\Pi^{(k+1)}\mathcal D_{2^{k}-1}\bn u)_{:,J}$ as: 
	\begin{align}(\Pi^{(i+1)}\mathcal D_{2^{i}-1}\bn u)_{:,J}&=\frac{1}{\Mycomb[k+1]{j}}\sum_{\Set{r}{s(b(r))=j}}(\mathcal D_{2^{i}-1}\bn u)_{:,r} (\text { let }j=s(b(J)))\\&=\frac{1}{\Mycomb[k+1]{j}}\Big[\sum_{\Set{r \text{ even}}{s(b(r))=j}}\frac{\bn D_x \bn p_{s(b(r))}}{\sqrt{\Mycomb[k]{s(b(r))}}}+\sum_{\Set{r \text{ odd}}{s(b(r))=j}}\frac{\bn D_y \bn p_{s(b(r))-1}}{\sqrt{\Mycomb[k]{s(b(r))-1}}}\Big]\\&=\frac{1}{\Mycomb[k+1]{j}}\Big[\frac{\bn D_x \bn p_{j}}{\sqrt{\Mycomb[k]{j}}}\sum_{\Set{r \text{ even}}{s(b(r))=j}}1+\frac{\bn D_y \bn p_{j-1}}{\sqrt{\Mycomb[k]{j-1}}}\sum_{\Set{r \text{ odd}}{s(b(r))=j}}1\Big]
	\end{align}
	Now, summation of 1 over any set is same as  the number of elements in that set. Let $Crd(S)$ denote the number of elements in S. First we compute,
	$Crd(\Set{r\in [2^{k+1}]}{\text{r odd, } s(b(r))=j})=Crd(\Set{2l+1\in [2^{k+1}]}{l\in [2^k], s(b(2l+1))=j})$.  Using $s(b(2l+1))=s(b(l))+1$  we get:$Crd(\Set{2l+1\in [2^{k+1}]}{l\in [2^k], s(b(2l+1))=j})=Crd(\Set{l\in [2^k]}{s(b(l))=j-1}).$ Now, $Crd(\Set{l\in [2^k]}{s(b(l))=j-1})$ is the number of non-negative integers less than $2^k$ which have $j-1$ ones in their binary code. Therefore, $Crd(\Set{l\in [2^k]}{s(b(l))=j-1})=\Mycomb[k]{j-1}.$ Similarly, $Crd(\Set{r \text{ even}}{s(b(r))=j})=\Mycomb[k]{j}.$
	Plugging these we get:$(\Pi^{(i+1)}\mathcal D_{2^{i}-1}\bn u)_{:,J}=\frac{1}{\sqrt{\Mycomb[k+1]{j}}}\Big[\frac{\sqrt{\Mycomb[k]{j}}}{\sqrt{\Mycomb[k+1]{j}}}\bn D_x \bn p_{j}+\frac{\sqrt{\Mycomb[k]{j-1}}}{\sqrt{\Mycomb[k+1]{j}}}\bn D_y \bn p_{j-1}\Big]$.
	Now, $(\Pi^{(k+1)}\mathcal D_{2^k-1}-\bn v)_{:,J}=\frac{1}{\sqrt{\Mycomb[k+1]{j}}}\Big[\frac{\sqrt{\Mycomb[k]{j}}}{\sqrt{\Mycomb[k+1]{j}}}\bn D_x \bn p_{j}+\frac{\sqrt{\Mycomb[k]{j-1}}}{\sqrt{\Mycomb[k+1]{j}}}\bn D_y \bn p_{j-1}-\bn q_{j}\Big].$ Comparing with the expressions of $\mathcal A_k\mathcal D_k$ and using the fact that there are $\Mycomb[k+1]{j}$ columns with $s(b(\cdot))=j$ we get the result.
\end{proof}

\eqtgv*
\begin{proof}
	Denote $S_D(\bn g)=\Set{{\sum_{i=0}^{n-1}}\alpha_{n-i-1}\|\Pi^{(i+1)}\mathcal D_{2^i-1}\bn u_i-\bn u_{i+1}\|_{1,2}}{\bn u_n=\bn 0,\bn u_0=\bn g ,\bn u_i\in \mathcal R(\Pi^{(i)})},$ and $S_C(\bn g)=\Set{{\sum_{i=0}^{n-1}}\alpha_{n-i-1}\|\mathcal A_i\mathcal D_{i}\bn p_i-\bn p_{i+1}\|_{1,2}}{\bn p_n=\bn 0,\bn p_0=\bn g ,\bn p_i\in \mathbb R^{N\times{(i+1)}}}.$ We prove the above result by showing that $S_D(\bn g)=S_C(\bn g).$ First we show that $S_C(\bn g)\subseteq S_D(\bn g).$ Consider any element $e_C$ of $S_C(\bn g).$ Then $$e_C=\sum_{i=0}^{n-1}\alpha_{n-i-1}\|\mathcal A_i\mathcal D_{i}\bn p_i-\bn p_{i+1}\|_{1,2}\text { where }{\bn p_n=\bn 0,\bn p_0=\bn g ,\bn p_i\in \mathbb R^{N\times{(i+1)}}\text{ for }i =1,...,n-1}.$$ Choose $(\bn u_i)_{:,j}=\frac{(\bn p_i) _{:,s(b(j))}}{\sqrt{\Mycomb[i]{s(b(j))}}}$ for all $i\in [n+1]$ and all $j\in [2^i]$. Then, by \cref{ch4:lemma:eq_cost} we have that $e_C=\sum_{i=0}^{n-1}\alpha_{n-i-1}\|\Pi^{(i+1)}\mathcal D_{2^i-1}\bn u_i-\bn u_{(i+1)}\|_{1,2}.$ Therefore, $e_C\in S_D(\bn g).$ Now, we show $S_D(\bn g)\subseteq S_C(\bn g).$ Consider any element $e_D\in S_D(\bn g).$ Then, $e_D=\|\Pi^{(i+1)}\mathcal D_{2^i-1}\bn u_i-\bn u_{i+1}\|_{1,2}$ where $\bn u_n=\bn 0,\bn u_0=\bn g \text{ ,and }\bn u_i\in \mathcal R(\Pi^{(i)})$ for $i=1,..,n-1.$As $\bn u_i\in \mathcal R(\Pi^{(i)})$ for all $i\in [n]$, we have $(\bn u_i)_{:,m}=(\bn u_i)_{:,o}$ if $s(b(m))=s(b(o)).$ Therefore, there exists $\bn q_i$ such that $\bn (u_i)_{:,j}=\frac{(\bn q_i) _{:,s(b(j))}}{\sqrt{\Mycomb[i]{s(b(j))}}}$ for all $i\in [n+1]$ and all $j\in [2^i]$. Therefore, by \cref{ch4:lemma:eq_cost} $e_D$ is also in $S_C(\bn g).$\end{proof}

\appendix

\section{References}
\bibliography{References}{}
\bibliographystyle{iopart-num}

\end{document}